\documentclass{article}[11pt]

\usepackage{amsmath}
\usepackage{amssymb}
\usepackage{amsthm}

\usepackage{amsopn}

\usepackage{lipsum}
\usepackage{amsfonts}
\usepackage{graphicx}
\usepackage{epstopdf}
\usepackage{algorithmic}
\usepackage{xcolor}
\ifpdf
  \DeclareGraphicsExtensions{.eps,.pdf,.png,.jpg}
\else
  \DeclareGraphicsExtensions{.eps}
\fi

\newfont{\BB}{msbm10 at 10pt} 
\newcommand{\CBB}{\mbox{{\BB C}}}  

\newtheorem{lemma}{Lemma}
\newtheorem{theorem}{Theorem}
\newtheorem{definition}{Definition}
\newtheorem{proposition}{Proposition}

\begin{document}

\title{Proof of non-convergence of the short-maturity expansion for the SABR model}
\author{Alan L. Lewis\footnote{Newport Beach, California, USA; email: alewis@financepress.com}
  \quad and Dan Pirjol\footnote{School of Business, 
  	Stevens Institute of Technology, Hoboken, NJ 07030; email:dpirjol@gmail.com}}

\date{26 July 2021}

\maketitle
	
\begin{abstract}
We study the convergence properties of the short maturity expansion of option 
prices in the uncorrelated log-normal ($\beta=1$) SABR model. 
In this model the option time-value can be represented as an integral of the form
$V(T) = \int_{0}^\infty e^{-\frac{u^2}{2T}} g(u) du$ with $g(u)$ a ``payoff 
function'' which is given by an integral over the McKean kernel $G(s,t)$. 
We study the analyticity properties of the function $g(u)$ in the complex
$u$-plane and show that it is holomorphic in the strip $|\Im(u) |< \pi$. Using this result
we show that the $T$-series expansion of $V(T)$ and implied volatility are asymptotic (non-convergent for any $T>0$). 
In a certain limit which can be defined either as the large volatility limit $\sigma_0\to \infty$ at fixed $\omega=1$, or the small vol-of-vol limit $\omega\to 0$ limit at fixed $\omega\sigma_0$, the short maturity $T$-expansion for the implied volatility 
has a finite convergence radius $T_c = \frac{1.32}{\omega\sigma_0}$. 
\end{abstract}



\section{Introduction and motivation}
\label{sec:1}

The SABR model is a versatile stochastic volatility model which has proved
very popular with practitioners since its introduction almost 20 years ago \cite{SABR}.
It was originally introduced to model interest rate volatilities, 
but its application has been extended later also to other asset classes, 
such as FX and commodities. 
The model is described by the diffusion
\begin{eqnarray}
&& dS_t = \sigma_t \mathcal{C}(S_t) dW_t \\
&& d\sigma_t = \omega \sigma_t dZ_t
\end{eqnarray}
where $(W_t, Z_t)$ are standard Brownian motions correlated with correlation 
$\rho\leq 0$.
The volatility of volatility (vol-of-vol) parameter $\omega$ determines the curvature of the implied volatility, and the backbone function $\mathcal{C}(S_t)$ is introduced such that the model captures the smile dynamics of the ATM (``at-the-money") implied volatility 
under spot price changes.

In the original SABR paper \cite{SABR} the backbone function was chosen 
as a power function (CEV-like)  $\mathcal{C}(S)=S^\beta$ with $0\leq \beta \leq1$. 
In practice more complicated forms are
used, such as the three-regime backbone of de Guillaume, Rebonato and Pogudin \cite{Rebonato}, reflecting the empirically observed backbone behavior of swaption volatilities. Since in the academic literature the SABR model is typically defined with CEV-like backbone, we also use the same convention and refer to the values of 
$\beta$ in the SABR model implied as the CEV backbone. 

The leading order in the short maturity expansion for the implied volatility for the SABR model was obtained in \cite{SABR}. 
The subleading $O(T)$ correction was also computed in this paper at the ATM point. The result has a simple analytical form, and is easily implemented in practice for model simulation and calibration. This feature contributed to the widespread adoption and popularity of the model. 

Higher order corrections to the short maturity expansion of the implied volatility in the SABR model were obtained by Henry-Labord\'ere \cite{HLbook} and Paulot \cite{Paulot2009}. The complete $O(T^2)$ contribution
was obtained in \cite{Paulot2009}, although its evaluation involves numerical integration for some terms. 
A systematic algorithm for expanding the implied volatility in a double series expansion in log-strike $x=\log(K/S_0)$ and maturity $T$ was mentioned in \cite{Lewis2017}, and is used here once to generate (\ref{ATMfull}) below. (However, none of our subsequent results rely upon this algorithm.)

Throughout, we work with the so-called ``log-normal" SABR model: $\beta=1$. 
The (Black-Scholes) implied volatility in this model has the
 full parametric dependence $\sigma_{\rm BS} = \sigma_{\rm BS}(x,T,\sigma_0,\omega,\rho)$.
 Here $x=\log(K/S_0)$, so ATM means $x=0$, and $T$ is the  time to option maturity. 
 
 For reference purposes we give here the expansion of the ATM implied volatility to $O(T^2)$ (to our knowledge the full $O(T^2)$ term is new \cite{unpublished}): 
\begin{eqnarray} 
&& \hspace{30pt} \frac{1}{\sigma_0} \sigma_{\rm BS}(0,T,\sigma_0,\omega,\rho) =
1 + \frac{1}{24} \sigma_0 \omega T
\Big[6 \rho  + \frac{\omega}{\sigma_0} (2- 3 \rho^2) \Big]  \label{ATMfull}\nonumber \\
&&  \,\, + \frac{1}{1920} \omega^2 \sigma_0^2 T^2 
\Big[(-80 + 240 \rho^2 ) +
 \frac{\omega}{\sigma_0} \rho (240 - 180 \rho^2)  +
\frac{\omega^2}{\sigma_0^2} (-12 + 60 \rho^2 - 45 \rho^4) 
\Big] \nonumber \\
&& \qquad +  O(T^3) \,. 
\end{eqnarray}

From here on, we work with $\omega=1$; 
the general case can be recovered as

\begin{equation}\label{omscaling}
	\sigma_{\rm BS}(x,T,\sigma_0,\omega,\rho) = 
	\omega \times \sigma_{\rm BS} \left( x,\omega^2 T,\frac{\sigma_0}{\omega},1,\rho \right).
\end{equation}
With $\omega=1$ and $x=\rho=0$ (suppressing the display of all three parameters), a second reference expansion is:
\begin{eqnarray} \label{CapSigdef}
	& \quad \Sigma_{\rm BS}^2(T,\sigma_0) \equiv \left(\frac{\sigma_{\rm BS}}{\sigma_0} \right)^2  =
	1 + \frac{1}{6} T - \frac{1}{180} T^2 (1 + 15 \sigma_0^2) + \frac{1}{1680} T^3 (4 - 161 \sigma_0^2)
	 \label{ATMfull2} \\
	&  \,\, - \frac{1}{453600} T^4 (579 + 29980 \sigma_0^2  - 7560 \sigma_0^4) + O(T^5). \nonumber
\end{eqnarray}

Although efficient numerical techniques are available for option pricing with general strike and maturity \cite{Lewis2017,AKSbook}, series expansions are very convenient. Thus, it is useful to explore their limits of applicability. We address the following questions:  what is the nature of the short-maturity expansion of the implied volatility around the ATM point? Is it strictly asymptotic or convergent? If convergent, is there a finite radius of convergence? 

In this note, we show that the full expansion underlying
(\ref{ATMfull2}) is strictly asymptotic, via a careful analysis of the (closed-form) SABR option value $V(T,\sigma_0)$. Briefly, the argument is as follows. We show in (\ref{CATM}) below that the value function can be put into the form
$V(T,\sigma_0) = C e^{-T/8} \int_0^{\infty} e^{-u^2/2T} g(u,\sigma_0) \, \frac{du}{\sqrt{T}}$, calling $g$ a ``payoff function''. 
Our key result, Theorem~\ref{thm:1}, establishes that
$g(u,\sigma_0$) admits an analytic continuation to the strip $|\Im(u)| < \pi$ in the complex $u$-plane. This function has singularities at $u = \pm i \pi$. 
Theorem~\ref{thm:1} may be of separate interest because $g(u,\sigma_0)$ itself is given by a non-trivial integral. From Theorem~\ref{thm:1}, the non-convergence of (\ref{ATMfull2}) for all $T > 0$ follows from term-by-term integration of $g$'s power series in $u$.

Arguably, non-convergence may have been the expected result. However,  
there is a curious and interesting ``large $\sigma_0$ scaling limit" -- where the convergence story changes.  This limit was  introduced and first studied in \cite{PES}.
 In that limit, under our current setup, take $\sigma_0 \rightarrow \infty$ 
and $T \rightarrow 0$, holding $\tau \equiv \frac{1}{2} \sigma_0 T$ fixed. 
(Under general $\omega$ this limit corresponds to taking $\omega \to 0, \sigma_0\to \infty$ at fixed and arbitrary $T$, holding $\sigma_0 \omega$ fixed.)
Substituting  $T = 2 \tau/\sigma_0$ in  (\ref{ATMfull2}), one sees that 

\begin{equation} 
	 \hat{\Sigma}_{\rm BS}^2(\tau) \equiv \lim_{\sigma_0 \rightarrow \infty}
	  \Sigma_{\rm BS}^2 \left( \frac{2 \tau}{\sigma_0},\sigma_0 \right) = 
	 1 - \frac13\tau^2 + \frac{4}{15} \tau^4 - \frac{92}{315}
	\tau^6 + O(\tau^8). \label{Scaled}
\end{equation}
Mechanically, large $\sigma_0$ scaling has the effect of (i) suppressing all the odd $T$-powers in (\ref{ATMfull2}) and
(ii) keeping only the largest power of $\sigma_0$ in each even $T$-power.  
In the Appendix, we present a (new) derivation of the closed-form relation for $\hat{\Sigma}_{\rm BS}^2(\tau)$,
a relation first obtained in \cite{PES} by a time discretization argument. 
The series in (\ref{Scaled}) has a finite (non-zero) convergence radius in the complex $\tau$-plane (see Sec. \ref{sec:6}).

\section{The SABR value function}
\label{sec:2}

Our starting point is the exact representation of the (time) value $V$ of call option prices in the uncorrelated $\beta=1$ SABR model. From Eq.~(3.103) in \cite{AKSbook}:
\begin{eqnarray}\label{3.103}
&& V(K,T) \equiv \mathbb{E}[(S_T - K)^+] - (S_0 - K)^+ \\
&& \quad  = \frac{2\sqrt{K S_0}}{\pi}
\int_{s_-}^\infty 
\frac{G(T,s)}{\sinh s} \sin \left( \frac12\sigma_0 \sqrt{\sinh^2 s- \sinh^2 s_-}\right) ds \nonumber
\end{eqnarray}
with $s_- = \frac{1}{\sigma_0} \log|K/S_0|$. 
This expression assumes $\omega=1$. This parameter can be restored to a general value by replacing $T\to \omega^2 T$ and $\sigma_0 \to \sigma_0/\omega$ as in 
(\ref{omscaling}).

The function $G(t,s)$ is related to the McKean heat kernel $\mathcal{G}(t,s)$
\cite{McKean1970} for the Brownian motion on the Poincar\'e hyperbolic plane
$\mathbb{H}^2$. The precise relation is 
$G(t,s) = 2\pi \int_s^\infty \mathcal{G}(t,u) \sinh u du$.
The function $G(t,s)$ is given explicitly by
\begin{equation}\label{3.83}
G(t,s) = \frac{e^{-t/8}}{\sqrt{\pi t}} \int_s^\infty \frac{e^{-u^2/(2t)} \sinh u}{
\sqrt{\cosh u - \cosh s}} du \,.
\end{equation}
Geometrically, suppose $s(x,y)$ is the hyperbolic distance between two points $x$ and
$y$ in $\mathbb{H}^2$. 
Then,  $\mathcal{G}(t,s(x,y))$ is the transition density for a hyperbolic Brownian particle,
starting at $x$, to reach the point $y$ after time $t$. Thus $G(t,s)$ is similar to a
complementary distribution function – the tail probability for the particle to move a hyperbolic distance greater than $s$ at time $t$. (See \cite{AKSbook} for further discussion.)

A similar integral expression to (\ref{3.103}) holds in the uncorrelated
SABR model with $0<\beta<1$, with different upper and lower integration limits,
and a different form for the sine factor. 
We study here the $\beta=1$ case for definiteness.

Combining (\ref{3.103}) with the integral representation of $G(t,s)$ in 
Eq.~(\ref{3.83}) and exchanging the order of 
integration, the option time value can be put into the form
\begin{equation}\label{Cgeneral}
V(K,T) = \frac{2\sqrt{K S_0} e^{-T/8}}{\pi^{3/2} \sqrt{T} }
\int_{s_-}^\infty e^{-\frac{u^2}{2 T}} g(u, s_- ) du
\end{equation}
with
\begin{eqnarray}\label{galt}
&& g(u,s_- ) \equiv \sinh u \int_{s_-}^u \frac{1}{\sqrt{\cosh u - \cosh s}} h(s,s_-) ds\,,\\
&& h(s,s_- ) \equiv 
\frac{\sin(\frac{\sigma_0}{2} \sqrt{\sinh^2 s - \sinh^2 s_-})}{\sinh s}\,.
\end{eqnarray}

At the ATM point $K=S_0$ we have $s_-=0$ and the expression (\ref{Cgeneral}) simplifies as
\begin{equation}\label{CATM}
 V(K=S_0,T) = \frac{2S_0 e^{-T/8}}{\pi^{3/2} \sqrt{T} }
\int_0^\infty e^{-\frac{u^2}{2 T}} g(u) du
\end{equation}
with
\begin{equation}\label{galtdef}
g(u) \equiv \sinh u \int_0^u \frac{1}{\sqrt{\cosh u - \cosh s}} h(s) ds\,,\quad
h(s) \equiv \frac{\sin(\frac{\sigma_0}{2} \sinh s)}{\sinh s}, \quad (u > 0).
\end{equation}
We study the analyticity of $V(K,T)$ in the complex $T$ variable. 
It is instructive to first generalize the situation.

\section{Analyticity of a class of general value functions $V(T)$}

Consider the analyticity of general value functions $V(T)$ in the complex $T$-plane
which can be represented in the form 
\begin{equation}\label{Vdef}
V(T) = c_1 e^{-c_2 T} \int_0^\infty e^{-\frac{u^2}{2T}} g(u) \frac{du}{\sqrt{2\pi T}},
\end{equation}
and where $g(u)$ is any ``analytic payoff function". Here $c_{1,2}$ are two model-dependent constants,  
irrelevant for analyticity. We introduce the following definition.

A payoff function $g(u)$ is said to be an analytic payoff function if there exists a function $G(z)$ of the
complex variable $z$ which is regular (analytic and single-valued) in the circle $|z| < R$ $(0 < R \le \infty)$ and a constant
$\Delta > 0$ such that $G(u) = g(u)$ for $0 \le u < \Delta$. 
We call $G(z)$ the analytic continuation of $g(u)$.\footnote{\label{ft:Lukacs} Our definition is in the spirit of Lukacs \cite{Lukacs} for ``analytic characteristic functions". A difference is that the agreement between $G$ and $g$ need only hold here for (an interval of) the \emph{positive} real axis.}  When $R = \infty$, then $G(z)$ is entire.

A further adopted restriction, very convenient for our problem, assumes $G(u)$ an odd function:
$G(-u) = -G(u)$. As it turns out, this antisymmetry holds for analytic continuations $G(u)$ under both
SABR and Black-Scholes (BS) models, with $G_{\rm {SABR}}(u)$ and 
$G_{\rm {BS}}(u)$ respectively. There is an important (subtle) point regarding oddness. As it's seen in the BS case
in an elementary way, consider that first. 
 
 With $x_T = \log S_T$, a generic call option value function $V(T) = E[(e^{x_T} - K)^+ | I_0]$, with $K$ the strike price,
 using $x^+ \equiv \max(x,0)$. Here $I_0$ is the set of initial conditioning information: $S_0$ in the BS model,
 $(S_0,\sigma_0)$ in the SABR model, etc. 
  Under BS, $S_T$ follows geometric Brownian motion
 with $x_T - x_0 \sim N(-\frac{1}{2} \sigma^2 T, \sigma^2 T)$. Here $N(\mu,v)$ denotes a normal distribution with mean $\mu$ and variance $v$, and (for this section alone) ``$\sim$" denotes ``is distributed as". Taking ATM with $S_0 = K = \sigma =1$, 
 
 \begin{eqnarray}\label{BSpayoff}
 	& V_{BS}(T) \,\, &= \int_0^\infty e^{-(u + \frac{1}{2} T)^2/2 T} (e^u -1) \frac{du}{\sqrt{2\pi T}}	 \\ 
 	 && = 2 \, e^{-\frac{1}{8} T} \int_0^\infty e^{-\frac{u^2}{2 T}}  
  \sinh \left( \frac{u}{2} \right) \frac{du}{\sqrt{2\pi T}} = \mbox{Erf} \left( \frac{T^{1/2}}{2^{3/2}} \right),  \nonumber
 \end{eqnarray}
showing a well-known result using the error function. Thus, $G_{BS}(u) = \sinh \left( \frac{u}{2} \right)$, odd as advertised. But, of course the ``original" payoff 
function $g_{\rm BS}(u)$ was defined for $u < 0$ ($S_T < K$) and vanishes there. Arguably,
$g_{\rm BS}(u) = \{ (\sinh(\frac{u}{2}))^+: u \in \bf{R}\}$,  certainly not odd. 
Yet, given the representation (\ref{Vdef}), we find $g(u)$
admits an \emph{analytic continuation} $G(u)$, an odd function. Thus, $G_{\rm BS}(u) \neq g_{\rm BS}(u)$ for real negative $u$,
a possibility already hinted at in footnote \ref{ft:Lukacs}. 
 
We belabor the point because a similar thing happens with the SABR model. Manifestly from (\ref{galtdef}), $g(-u) = g(u)$, easily confirmed from a plot. But, the power series (\ref{eq:gSABRseries}) below is an odd series. The resolution of the apparent discrepancy is that the analytic continuation of
$g(u)$, to the negative real axis via $G(u)$, is \emph{not} found by mechanically taking a negative value of $u$ in the
integral of (\ref{galtdef}), even though that (extended) integral technically exists. (Indeed, a mechanical plot of 
(\ref{galtdef}) over an interval $(u_1,u_2)$, with $u_1 < 0 < u_2$, shows a function not even differentiable at $u=0$). Instead, the analytic continuation \emph{enforces the antisymmetry} of the power series for $g(u)$. That should motivate part of our 
key Theorem \ref{thm:1} below. That theorem will establish that $G_{\rm SABR}(u)$ is indeed an analytic payoff function
with finite convergence radius $R$.

\textbf{Small-maturity expansion of the value function.}
Under our assumptions, $G(u) = \sum_{k=0}^{\infty} a_k u^{2 k+1}$ for some 
sequence of coefficients
$\{a_k\}$, as long as $|u| < R \le \infty$. We allow finite $R$, or $R = +\infty$ for entire functions.
Integrating term-by-term gives a formal expansion of a ``normalized value function" $\hat{V}(T)$,
defined by 
\begin{eqnarray}\label{Vexp2}
	&& \int_0^\infty e^{-\frac{u^2}{2T}} \, G(u) \frac{du}{\sqrt{2\pi T}} = 
	\sum_{k=0}^\infty a_k \int_0^\infty e^{-\frac{u^2}{2T}} u^{2k+1} \frac{du}{\sqrt{2\pi T}}\\
	&& = 
	\sqrt{\frac{T}{2\pi}} \sum_{k=0}^\infty a_k \, (2T)^k \, \Gamma(1+k) =  \sqrt{\frac{T}{2\pi}} \, \hat{V}(T). \nonumber
\end{eqnarray}
$\hat{V}(T)$ differs from $V(T)$ by a $\sqrt{T}$ and the pre-factors $c_1 e^{-c_2 T}$ of (\ref{Vdef}).
Our issue is the convergence or not, of the power series for $\hat{V}(T)$. Consider three cases:

\begin{enumerate}
	\item $G(u)$ analytic with $R < \infty$. 
(Example: SABR with $R = \pi$. This will be shown below in Sec.\ref{sec:SABRanalyticity}.).
	\item $G(u)$ entire and of exponential type $k$. (Example: BS with $k = \frac{1}{2}$.).
	\item $G(u)$ entire of order 2 and type $k$.
\end{enumerate} 
Now, by the root test for 
convergence, if $\hat{V}(T) = \sum b_n T^n$ converges,
its radius of convergence is $r = \limsup_{n \rightarrow \infty} |b_n|^{-1/n}$.
Here $b_n = 2^n a_n \Gamma(1+n)$. Freely invoking  Stirling's approximation, 
we obtain the following convergence properties for each of the cases enumerated.

{\bf Case 1}. Since $G(u)$ is analytic with radius $R$, $\limsup_{n \rightarrow \infty} |a_n|^{-1/(2n)} = R$. Thus, the convergence radius of $\hat V(T)$ is
\begin{eqnarray*}\label{stirling1}
 &\sqrt{r} \,\, &=  \lim_{n \rightarrow \infty} |b_n|^{-1/(2n)} = 2^{-1/2} \times
    \lim_{n \rightarrow \infty} |a_n|^{-1/(2n)} |\Gamma(1+n)|^{-1/(2n)}  \\
 &   & =  2^{-1/2} \times
    \lim_{n \rightarrow \infty} R  \, e^{-\frac{1}{2} \log n - \frac{1}{2} + O((\log n)/n)} = 0,
\end{eqnarray*}
In words, the $T$-series for $\hat{V}(T)$, under Case 1, has zero radius of convergence. 

{\bf Case 2}. Since $G(u) = \sum_{k=0}^{\infty} a_k u^{2 k+1}$ and $G(u) =  O(e^{k u})$ at $\infty$, 
for large $n$, $|a_{n}| \sim k^{2 n}/(2 n!)$. Now
\begin{eqnarray*}\label{stirling2}
	&\sqrt{r} \,\, &=  \lim_{n \rightarrow \infty} |b_n|^{-1/(2n)} \sim \frac{1}{\sqrt{2} \, k} \times
	\lim_{n \rightarrow \infty} \left( \frac{\Gamma(1+n)}{\Gamma(1 + 2 n)} \right)^{-1/(2n)}  \\
	&   &=  \frac{1}{\sqrt{2} \, k}\times
	\lim_{n \rightarrow \infty} \, e^{\frac{1}{2} \log n + (\log 2 - \frac{1}{2}) + O((\log n)/n)} = +\infty,
\end{eqnarray*}
Thus, under Case 2, $\hat{V}(T)$ is entire. This agrees with the BS result from (\ref{BSpayoff}): since $\mbox{Erf}(z)$ is entire and odd, $\hat{V}(T) = \mbox{Erf}(c \sqrt{T})/\sqrt{T}$ is an entire function of $T$.

{\bf Case 3}. If $G(u) =  O(e^{-k u^2})$ at $\infty$, for large $n$, $|a_{n}| \sim k^n/n!$. Now
\begin{eqnarray*}\label{stirling3}
	&\sqrt{r} \,\, &=  \lim_{n \rightarrow \infty} |b_n|^{-1/(2n)} \sim \frac{1}{\sqrt{2k}} \times
	\lim_{n \rightarrow \infty} \left( \frac{\Gamma(1+n)}{\Gamma(1 + n)} \right)^{-1/(2n)}  = \frac{1}{\sqrt{2k}}.
\end{eqnarray*}
For order 2, type $k$ payoffs, the $\hat{V}(T)$ series converges for $T < \frac{1}{2 k}$.

\section{Analyticity of the SABR payoff function} \label{sec:SABRanalyticity}
\label{sec:4}

In this section we study the extension of the function $g(u)$ defined by the integral
(\ref{galtdef}) to complex values of $u$.
The integral (\ref{galtdef}) is well defined along the real axis $\Re(u)>0$. 
We would like to construct a holomorphic function $G(u)$ which reduces to $g(u)$ along the positive real axis, and determine its maximal domain of holomorphicity.
The limitations of this domain are due to the singularities of the factor $\sqrt{\cosh(u) - \cosh(s)}$ in the denominator. Defining this factor as a single-valued function for complex $u$ requires some care in the choice of the branch cut of the square root. We will choose to define the square root with a cut along the real positive axis, and denote it as $(\sqrt{z})_+$. Specifically, if $\sqrt{z}$ denotes the standard square-root with a branch-cut along the negative
real axis, then

\begin{eqnarray}\label{possqrt}
	(\sqrt{z})_+ = \left\{
	\begin{array}{cl}
		\sqrt{z}  & \Im(z) \ge 0, \\
		-\sqrt{z} &  \Im(z) < 0. 
		\end{array}
	\right.
\end{eqnarray}
This choice is guided by the following lemma. 

\begin{lemma}\label{lemma:1}
The equation $\cosh u - \cosh(w u) = z$ with $w \in [0,1]$ and real $z > 0$ has no solutions in the half-strip $\Re(u) \geq 0, 0 < \Im(u) < \pi$.

\end{lemma}

\begin{proof}
Writing $u=x+iy$ we have $\cosh u - \cosh(w u) = r(x,y) + i s(x,y)$ with
\begin{eqnarray}
&& r(x,y) = \cos y \cosh x - \cos(wy) \cosh(wx) \\
&& s(x,y) = \sin y \sinh x - \sin (wy) \sinh (wx)
\end{eqnarray}

We distinguish the two cases:

i) $x=0$. We have $s(0,y)=0$ and $r(0,y)\leq 0$ for all $ 0 < y < \pi$, so the
equation $\cosh(u) - \cosh(wu) = z >0$ clearly does not have a solution.

ii) $x>0$. Fix $x$ and vary $y$ in $[0,\pi]$. 
We will show that if $s(x,y_0)=0$ has a zero at $y_0 \in (0,\pi)$, then
$r(x,y_0)<0$, which proves the statement of the lemma.

\textit{Step 1.  }
For $y\in [0,\frac{\pi}{2}]$ we have the lower bound 
\begin{equation}
s(x,y) > [\sin y - \sin(wy)] \sinh x > 0 \,, \mbox{ for } 0 < y  < \frac{\pi}{2}
\end{equation}
since $\sin y$ is increasing on $[0,\frac{\pi}{2}]$. 
We used here  $\sinh(wx)<\sinh x$ which holds for all $x>0$.
This implies that if $s(x,y)$ has a zero $y_0$ then it must lie in 
$[\frac{\pi}{2},\pi]$.

\textit{Step 2. }
For $y\in [\frac{\pi}{2},\pi]$, the function $r(x,y)$ is negative. This follows
from the upper bound 
\begin{equation}
r(x,y) < \cosh (wx) [\cos y - \cos(w y)] < 0 \,, \mbox{ for } \frac{\pi}{2} < y <\pi 
\end{equation}
since $\cos y$ is decreasing on $[0,\pi]$. 
We used here $\cosh x > \cosh (wx)$ and $\cos y < 0$.
\end{proof}

We rely subsequently upon two textbook results for analytic continuation. The first is the classic Schwarz reflection principle. 
The second is the analyticity of certain functions defined via an integration. 
For ease of reference we state the result below, in the formulation of Stein and Shakarchi \cite{SteinShakarchi}.

\begin{theorem}[Stein and Shakarchi \cite{SteinShakarchi}, Th.~5.4 Ch.~2]\label{thm:0}
   Let $F(z,w)$ be defined for $(z,w) \in \Omega \times [0,1]$, where $\Omega$ is an open set in $\CBB$.
   Suppose $F$ satisfies the following properties:
   
   {\rm (i)} $F(z,w)$ is holomorphic in $z$ for each $w$.
   
   {\rm (ii)} $F$ is continuous on $\Omega \times [0,1]$.  
      
   \noindent Then the function $f$ defined on $\Omega$ by $f(z) = \int_0^1 F(z,w) \, dw$ is holomorphic.
\end{theorem}

We shall construct a function $G(u)$ which is holomorphic
in the entire strip $- \pi < \Im(u) < \pi$, and which reduces to the function 
$g(u)$ defined by the integral in (\ref{galtdef}) along the real positive $u$ axis.

\begin{definition}\label{def:1}
First, using  $(\sqrt{z})_+$, define $G(u) \equiv G_{Q1}(u)$ in the interior 
and boundaries of the upper half-strip in the first quadrant
 $Q_1 = \{ u  | arg(u) \in [0,\pi/2] \cap \Im(u)<\pi \}$ as the integral
\begin{equation}\label{gQ1def}
G_{\rm Q1}(u) \equiv \sinh u \int_0^u \frac{h(s) ds}{(\sqrt{\cosh u - \cosh s})_+} \,, 
u \in Q_1 \,.
\end{equation}
Elsewhere in the $|\Im(u)| <\pi$ strip, the function $G(u)$ is defined by combined application of the odd symmetry property in $u$ and Schwarz reflection principle
\begin{eqnarray}\label{thmcase2}
G(u) \equiv \left\{
\begin{array}{ccl}
-(G_{Q1}(-u^*))^* & \,, & u \in 
   Q_2 = \{ u | arg(u)\in (\frac{\pi}{2}, \pi ] \cap \Im(u)<\pi \}\\
-G_{Q1}(-u) & \,, &  u \in 
   Q_3 = \{ u | arg(u)\in [-\pi,-\frac{\pi}{2} ) \cap | \Im(u)|< \pi \} \\
(G_{Q1}(u^*))^* & \,, &  u \in 
   Q_4 =  \{ u | arg(u)\in [- \frac{\pi}{2},0 ) \cap  |\Im(u)| < \pi \}\\
\end{array}
\right.\nonumber \\
\end{eqnarray}
\end{definition}

\begin{theorem}\label{thm:1}
	$G(u)$ is holomorphic (and hence analytic) in $|\Im(u)| < \pi$.
	\end{theorem}

\begin{proof}
The $s$-integral in (\ref{galtdef}) is well defined for positive real $u$. 
We would like to define an analytic continuation of this integral to the complex 
$u$-plane, in the strip $-\pi < \Im (u) < \pi$, which agrees with the original integral for positive real $u$.

With $(\sqrt{z})_+$, along the upper side of the cut (the positive real axis), \newline
$(\sqrt{\cosh(u) - \cosh(wu)})_+$ is positive and real, and reproduces the denominator in (\ref{galtdef}).

\begin{figure}[h]
\centering
\includegraphics[width=2.2in]{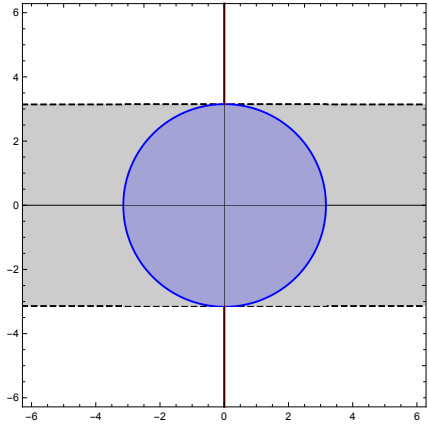}
\caption{The function $G(u)$ is holomorphic in the strip $-\pi < \Im(u) < \pi$. The convergence domain of the series expansion for $G(u)$ is the disc $|u| <\pi$, 
and the closest singularities to $u=0$ are the branch points at $\pm i \pi$.}
\label{Fig:strip}
\end{figure}

Next introduce the integral $G_{Q1}(u)$ defined as in (\ref{gQ1def})
for $u \in Q_1$. This can be written equivalently as
\begin{equation}
G_{Q1}(u) = 
u \sinh u \int_0^1 \frac{h(w u) dw}{(\sqrt{\cosh u - \cosh (wu)})_+} \,.
\end{equation}
The integrand is a holomorphic function for $u \in \mbox{Int}(Q_1)$  for each
$w\in [0,1]$ and is jointly continuous  in $u,w$.
Continuity follows by Lemma~\ref{lemma:1} which ensures that the argument of the square root never crosses the cut for all $u \in \mbox{Int}(Q_1)$.
This property is illustrated graphically in Figure~\ref{Fig:map}
which shows the mapping of the half-strip $0 \leq \Im (u) <\pi, \Re(u) \geq 0$ by the function $\cosh u - \cosh(w u)$ at fixed $w=0.8$. As $w$ is varied in $(0,1)$ the image of the half-strip does not cross the real positive axis, which ensures continuity of $(\sqrt{\cosh u - \cosh(wu)})_+$ in $u$. 

\begin{figure}[h]
\centering
\includegraphics[width=2.0in]{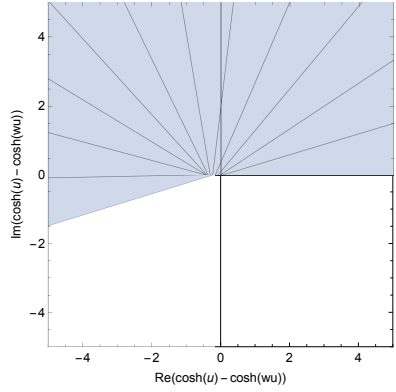}
\caption{The mapping of the half-strip $0\leq \Re (u), 0 \leq \Im (u) < \pi$
 by the function $\cosh(u) - \cosh(wu)$ with $w=0.8$. Lines of constant $\Im (u)$ are mapped to the radial curves extending from the center outwards. 
 }
\label{Fig:map}
\end{figure}

By Theorem \ref{thm:0} it follows that:

(i) the function $G_{Q1}(u)$ is holomorphic for $u\in {\rm Int}(Q_1)$. 

In addition, we have:
 
(ii) The limits of $G_{Q1}(u)$ along the axes bordering $Q_1$ are continuous with the interior values. Along the real $u$ axis this follows from the equality $(\sqrt{z})_+=\sqrt{z}$ for real positive $z$, and along the imaginary $u$ axis from Lemma~\ref{lemma:1}.

Combining (i) and (ii), the Schwarz reflection principle provides the analytic extension of $G(u)$ to $Q_1 \cup Q_4$. This gives the last line in (\ref{thmcase2}).


Enforcing the odd property $G(-u)=-G(u)$ and applying the Schwarz principle again provides 
the remaining continuations  to $Q_2$ and $Q_3$, and thus to the entire
open strip $-\pi < \Im(u) < \pi$.

\end{proof}

{\bf Comments.}
By construction, along the real positive $u$ axis, $G_{Q1}(u)$ reproduces the integral $g(u)$ in
(\ref{galtdef}). In addition, $G(u)$ approaches $g(u)$ continuously as $u$ approaches the real axis from below 
\begin{equation}\label{thmcase1}
	g(u) = \lim_{\epsilon\to 0^+} G_{Q1}(u - i\epsilon) \,,\quad
	\Im(u)=0, \, \Re(u)>0\,.
\end{equation} 
Of course, we more generally have that $G(u)$ is both continuous and continuously differentiable
as any axis is crossed in $|\Im(u)| < \pi$, since $G$ is analytic there. Along the negative real axis we have $G(-u) = - G(u)$ by the odd symmetry in $u$. All of this smoothness is evident in Fig. \ref{Fig:G}, which shows plots of both the real
and imaginary parts of $G(u)$ in a rectangle.

\begin{figure}[h]
\centering
\includegraphics[width=2.5in]{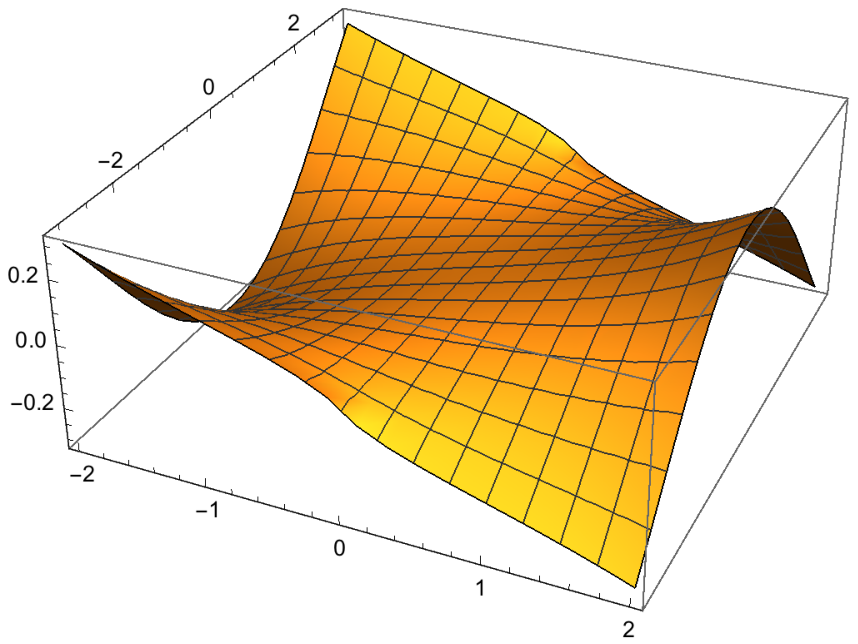}
\includegraphics[width=2.5in]{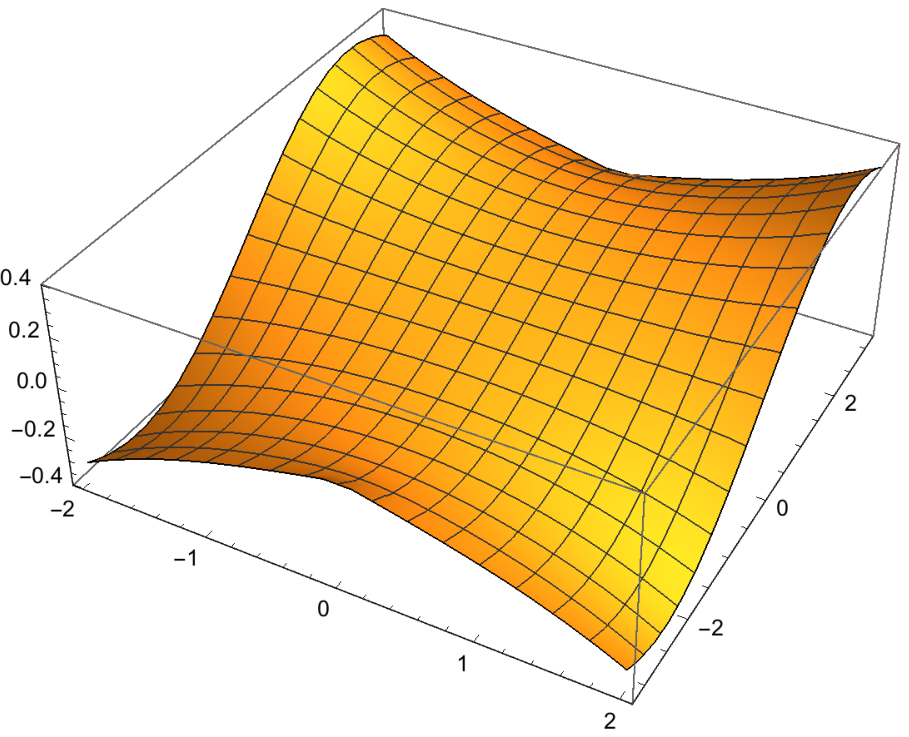}
\caption{Plots of $\Re (G(u))$ (left) and $\Im (G(u))$ (right) for $\sigma_0=0.1$,
with $u=x+iy$ in the ranges $x:[-2,2], y:[-\pi,\pi]$.}
\label{Fig:G}
\end{figure}

\subsection{The $G(u)$ power series}
By Theorem~\ref{thm:1}, the function $G(u)$ can be expanded in a power series around $u=0$ which is convergent for $|u|< \pi$. 
As advertised, the series contains only odd powers
\begin{equation}\label{gexp}
G(u) = \sigma_0 \sum_{k=0}^\infty a_k(\sigma_0) u^{2k+1}\,.
\end{equation}

The first few terms are
\begin{equation} \label{eq:gSABRseries}
G(u) = \frac{\pi}{2\sqrt2} \sigma_0 u \left( 1 + \frac{1}{48}(5 - \sigma_0^2) u^2 
+ \frac{23 - 110 \sigma_0^2 + 3 \sigma_0^4}{15360} u^4 + O(u^6) \right)\,.
\end{equation}

\subsection{Singularity structure} The function $G(u)$ has logarithmic
singularities on the imaginary axis of the form 
\begin{equation}
G(u) = 
\frac{\sigma_0}{2} \sqrt2 \cdot (u - u_\pm) \log\frac{1}{u-u_\pm} + \mbox{regular}\,,\quad u \to u_\pm = \pm i \pi \,.
\end{equation}

In order to study this singularity, take $u=iy$ along the imaginary
axis with $y<\pi$. Denoting the integration variable $s=it$ and neglecting the factor $h(s)$ which is regular as $s\to i\pi$, the integral is approximated as
\begin{equation}
\int_0^{y} \frac{i dt}{\sqrt{\cos y - \cos t}} = \int_0^y \frac{dt}{\sqrt{\cos t - \cos y}}
=\sqrt2 \log\frac{8}{\pi - y} + O(\pi-y)\,.
\end{equation}
The convergence domain of the Taylor series of $G(u)$ around $u=0$ is restricted by this singularity to the open disk $|u| < \pi$. See Fig.~\ref{Fig:strip}. 
The root test -- see Fig.~\ref{Fig:rt} -- confirms convergence of the series expansion
for $G(u)$ with a convergence radius $\pi$.

\begin{figure}[h]
\centering
\includegraphics[width=2.5in]{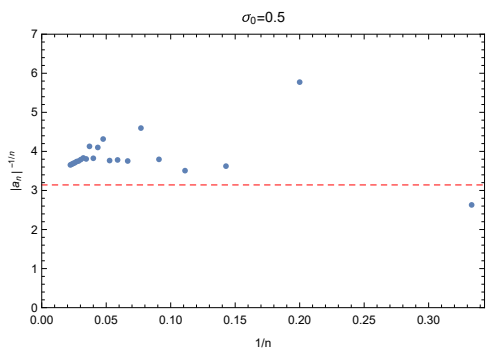}
\caption{Root test for the convergence of the series expansion (\ref{gexp})
for $G(u)$, showing the reduced coefficients $|a_n|^{-1/n}$ vs $1/n$. 
The horizontal line is at $\pi$.
$\sigma_0=0.5$.}
\label{Fig:rt}
\end{figure}

Application of ``transfer results" in Flajolet and Sedgewick (see equation (26) in Ch.VI.2 of \cite{Flajolet2008}) gives the leading large-$n$ order asymptotics of the coefficients in the series expansion (\ref{gexp})
\begin{equation}\label{asymcoefs}
a_k (\sigma_0) =  (-1)^{k} \frac{\sqrt2}{\pi^{2k} (2k)(2k+1) } + O(k^{-3}) \,.
\end{equation}

\vspace{0.5cm}

\subsection{Non-convergence (revisited)}

Recall from (\ref{Vdef}), dropping pre-factors:

\begin{equation}\label{Vexp}
V(T)= \int_0^\infty e^{-\frac{u^2}{2T}} g(u) \frac{du}{\sqrt{2\pi T}}  
 = \sqrt{\frac{T}{2\pi}} \sum_{k=0}^\infty a_k(\sigma_0) (2T)^k \Gamma(1+k).
\end{equation}
It's worth revisiting the non-convergence argument with the improved knowledge from (\ref{asymcoefs}). 
Now, the large-$k$ asymptotics of the coefficients have the form 
\begin{equation}\label{akasympt}
a_k (2T)^k k^k e^{-k} \sim \left( \frac{2k T}{\pi^2 e} \right)^k  \,,\quad
k \to \infty \,.
\end{equation}
Again the root test shows that the $T$-series for the ATM option price has zero radius of convergence.

\subsection{Non-convergence of the short-maturity expansion of the implied volatility function}
This is our title result. Consider the expansion for the implied volatility $\sigma_{\rm BS}(T) = \sum_{k=0}^\infty b_k T^k$. This is related to the ATM option price as\begin{equation}
\frac{1}{\sqrt{T}} C(K=S_0,T) = S_0 \frac{1}{\sqrt{T}}
\mbox{Erf}\left(\frac{\sigma_{\rm BS}(T) \sqrt{T}}{2\sqrt2}\right)\,.
\end{equation}

We prove that a finite convergence radius of the series for the implied variance $\sigma_{\rm BS}^2(T)$ implies a finite convergence radius of the option price 
$\frac{1}{\sqrt{T}} C(K=S_0,T)$. Since the latter series is non-convergent we conclude that the previous series must be non-convergent. 
The proof proceeds in two steps.

\textit{Step 1.} First we observe that the function $f(z)\equiv \frac{1}{\sqrt{z}} \mbox{Erf}( \sqrt{z})$ is entire. This follows from the application of the root test for convergence to its Taylor expansion
$f(z) = \frac{2}{\sqrt{\pi}}\sum_{n=0}^\infty (-1)^n \frac{1}{(2n+1) n!} z^n$. The root test gives that the convergence radius of this series expansion is infinite, which proves that $f$ is entire.

\textit{Step 2. }
Suppose $g(z)$ is an analytic function with finite convergence radius $|z|<R$ and denote $R_0>0$ the radius of the largest disk centered on $z=0$ which is mapped by $z \to g(z)$ to a region which does not include the origin.

Then $h(z) \equiv \frac{1}{\sqrt{z}} \mbox{Erf}[g(z) \sqrt{z}]$ has the convergence radius $\min(R,R_0)$. This follows from writing $h= \sqrt{g^2} \times f \circ (g^2 z)$ as the composition of $f(z)$ with $g^2(z) z$. The analyticity domain of $h$ is limited either by the analyticity domain of $g$, or by the branch cut of $\sqrt{g^2(z)}$ starting at the point where $g(z)=0$, and is thus the same as the disk $|z|<\min(R,R_0)$.

From the non-convergence of the series for $\frac{1}{\sqrt{T}} V(T)$, it follows that the series for $\sigma_{\rm BS}(T)$ also has zero convergence radius.

\section{Numerical illustrations and error estimates}
\label{sec:5}

The asymptotic nature of the $T$-expansion of option prices and implied volatility for the SABR model requires a careful application for practical use. 
The $T$-series (\ref{Vexp}) for the value function $V(T)$ must be truncated to some finite order $N$. Two issues must be addressed in relation to the use of asymptotic series: i) what is the optimal truncation order $N_*$, and ii) estimate the best attainable error of the series $\varepsilon_* = \inf_N|\varepsilon_N(T)|$,
where $\varepsilon_N(T) = V(T) - V_N(T)$ is the truncation error.

We illustrate these issues on the example of the value function $V_0(T)$ defined by taking $h(s)=1$. This situation corresponds to the small volatility regime $\sigma_0 \ll 1$, when the $h(s)$ factor is well approximated by a constant for 
$\sigma_0 \sinh s \ll 1$. 
For this case the integral in (\ref{galtdef}) can be evaluated exactly as
\begin{equation}\label{g0exact}
g_0(u) \equiv \sinh u \int_0^u \frac{ds}{\sqrt{\cosh u - \cosh s}} 
= -2i \sinh u \frac{F(\frac12 iu | -\mbox{cosech}^2(u/2))}{\sqrt{\cosh u-1}} \,,
\end{equation}
where $F(\phi| m) = \int_0^\phi (1 - m \sin^2\theta)^{-1/2} d\theta$ is the elliptic integral of the first kind.
The value function $V_0(T)$ is defined by (\ref{Vdef}) with the replacement 
$g(u) \to g_0(u)$.

Figure~\ref{Fig:2} shows the numerical evaluation of $V_0(T)$ from the series expansion (\ref{Vexp}) truncated to order $n$, plotted as a function of $n$, compared with numerical evaluation of $V_0(T)$ using the exact result (\ref{g0exact}) for the integrand. The different plots correspond to several values of the $T$ parameter.

\begin{figure}[h]
\centering
\includegraphics[width=2.5in]{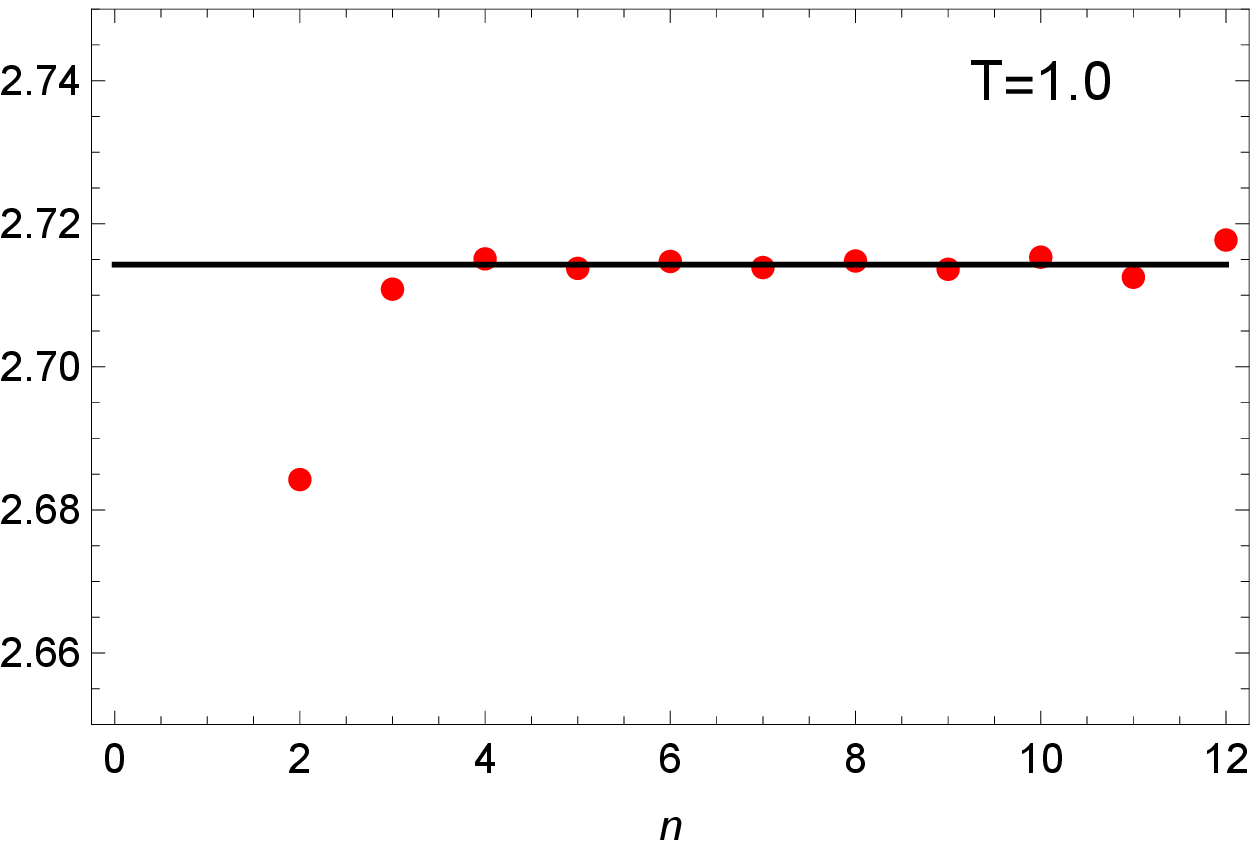}
\includegraphics[width=2.5in]{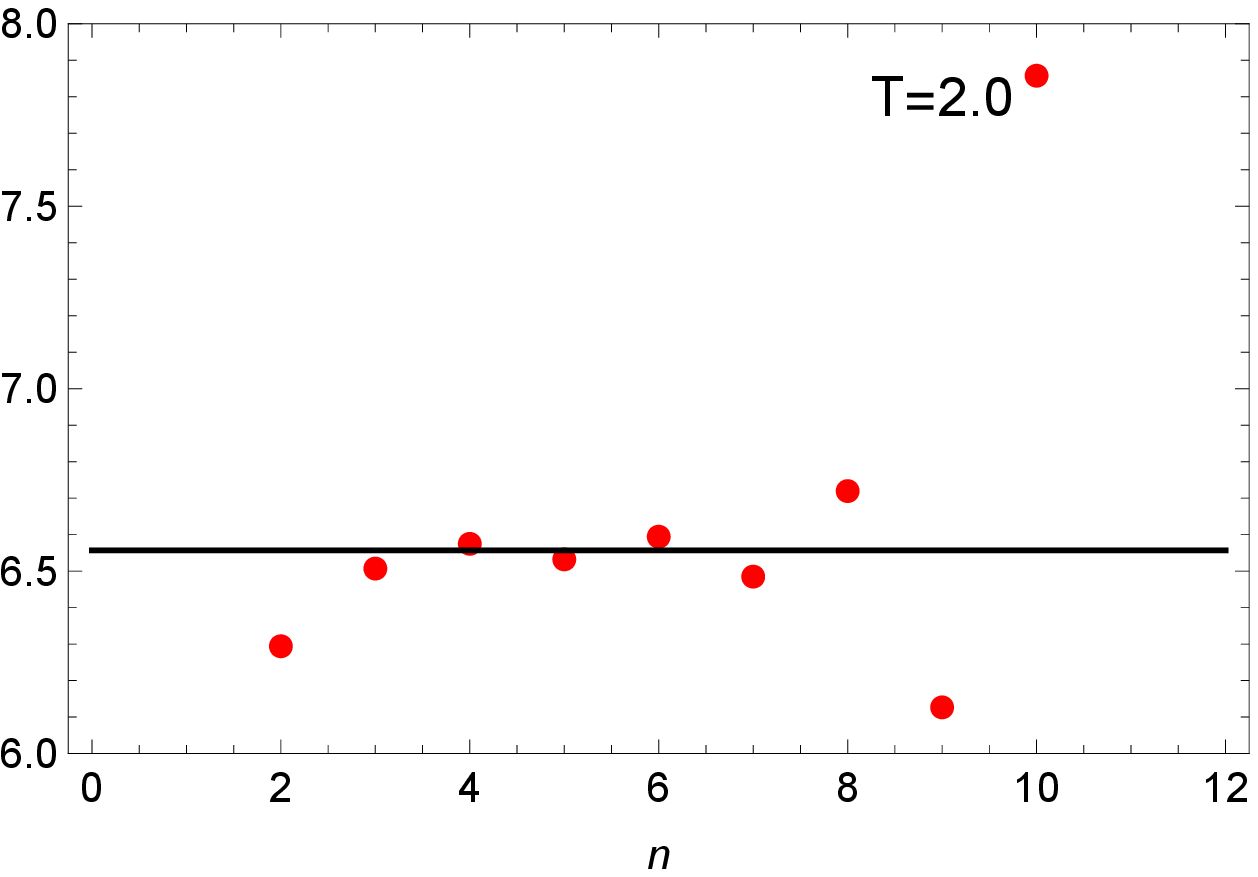}
\includegraphics[width=2.5in]{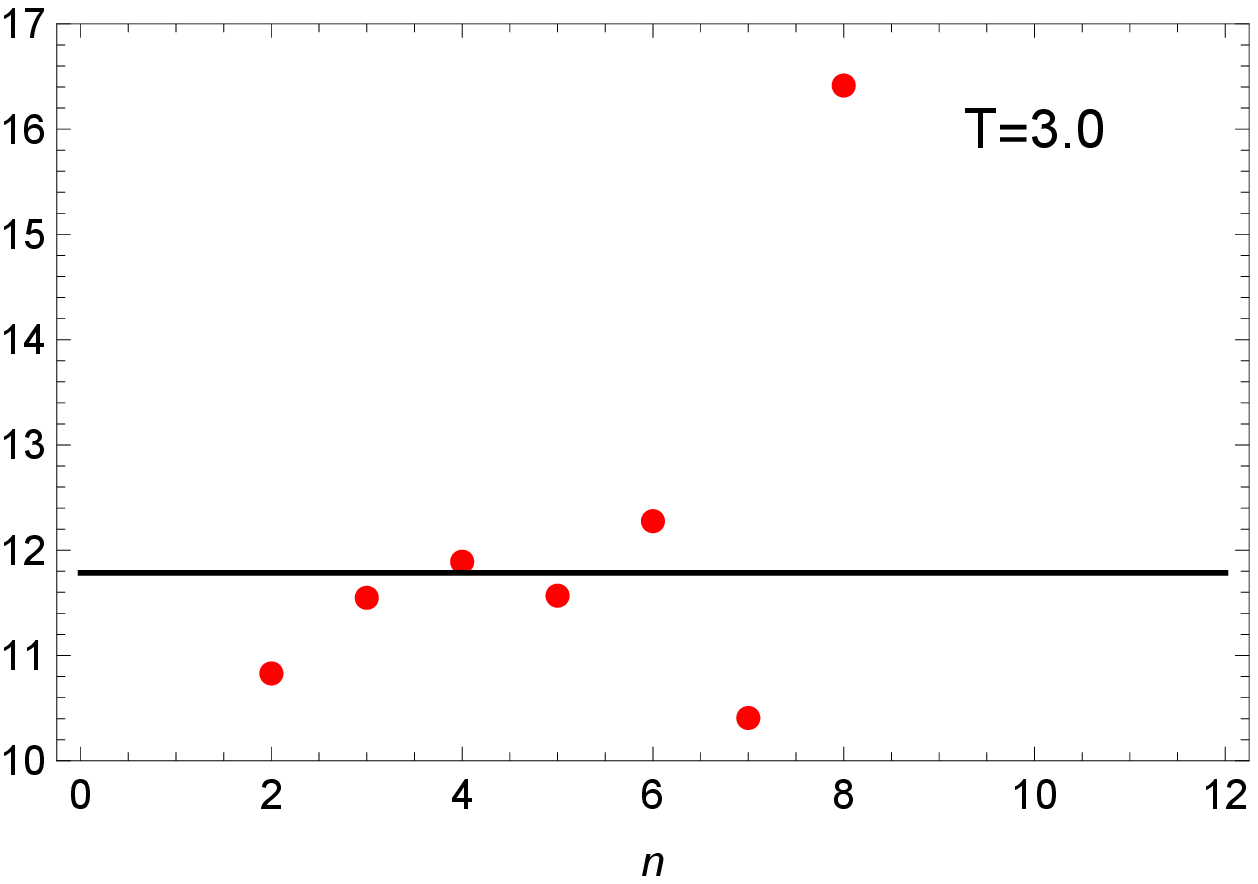}
\includegraphics[width=2.5in]{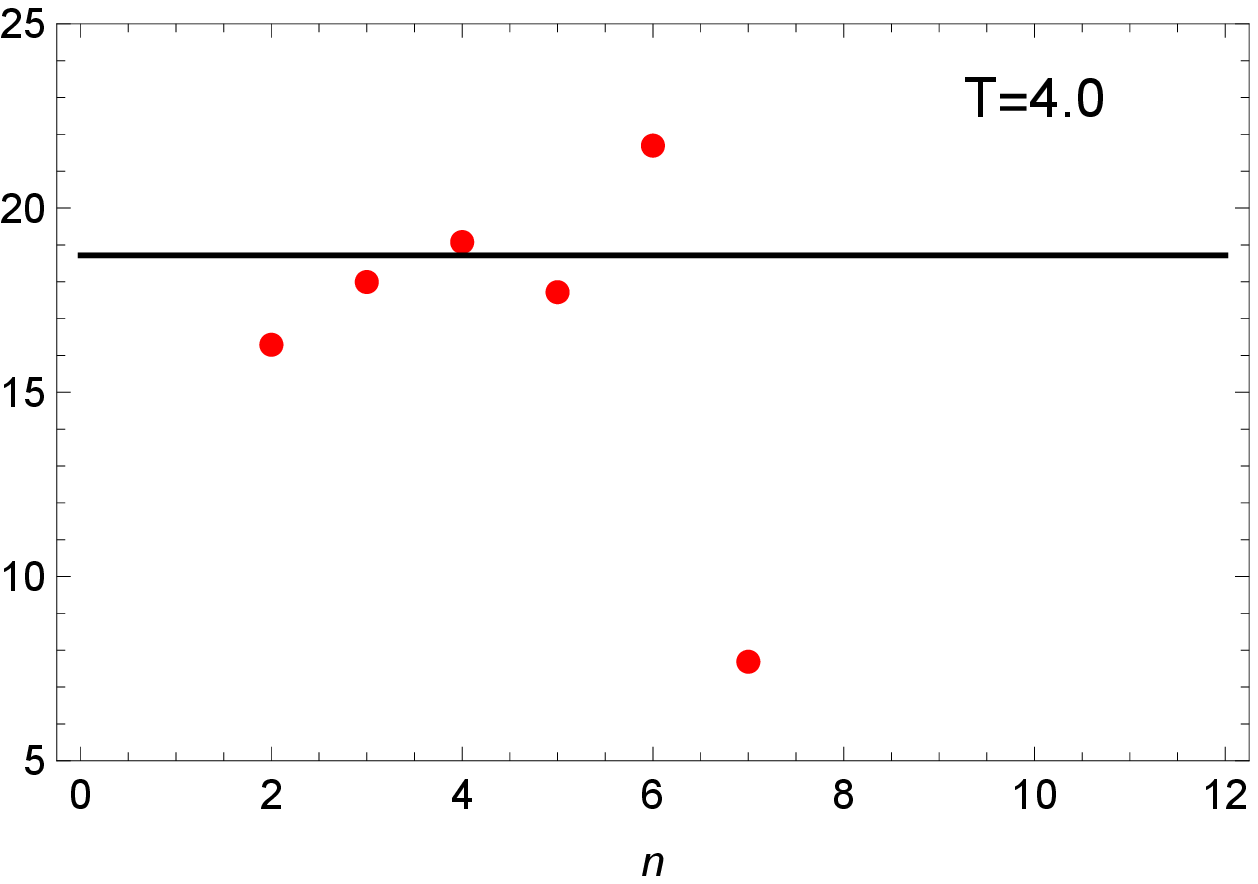}
\caption{Dots: The partial sum of $V_0(T)$ from the series expansion (\ref{Vexp}) keeping terms up to $O(T^n)$, vs $n$. 
Horizontal black line: numerical evaluation of $V_0(T)$.}
\label{Fig:2}
\end{figure}

We note from these plots that the truncated series agrees best with the numerical evaluation at that order $N_*$ where the last neglected term $V_{N_*+1}-V_{N_*}$ reaches a minimum. This agrees with the typical behavior of asymptotic series \cite{Boyd}. The optimal truncation order $N_*$ can be estimated from the large-order asymptotics (\ref{akasympt}) of the coefficients as 
$N_* \sim \frac{e\pi^2}{2T}$. 
$N_*$ decreases as $T$ increases, and approaches unity for 
$T\sim \frac12 e\pi^2 \simeq 13.4$. 
These arguments show that the asymptotic series has a maximum range of validity and breaks down for too large $T$.

An upper bound on the optimal truncation error of the series can be obtained from a bound on the contribution to the integral (\ref{Vdef}) from the region $u>\pi$. 
Denoting 
\begin{equation}
V_{\rm err}(T) = \int_\pi^\infty e^{-\frac{u^2}{2T}} g(u) \frac{du}{\sqrt{2\pi T}}
\end{equation}
we have 
\begin{equation}\label{Verrbound}
|V_{\rm err}(T) | \leq \sqrt2 \sigma_0 
\left\{ \sqrt{\frac{T}{2\pi} } e^{-\frac{\pi(\pi-T)}{2T}}
+ \frac12 e^{T/8} T N\left( \frac{\frac12 T - \pi}{\sqrt{T}} \right) \right\}\,.
\end{equation}
with $N(x) = \int_{-\infty}^x e^{-\frac12 t^2} \frac{dt}{\sqrt{2\pi}}$ the CDF of the standard normal distribution.

\begin{figure}[h]
\centering
\includegraphics[width=3.0in]{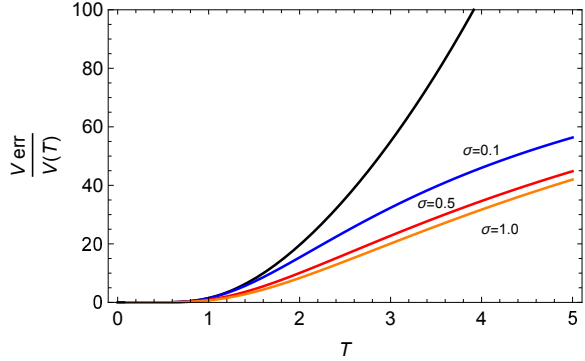}
\caption{The best attainable error (in percent), measured by the ratio of the contribution to $V(T)$ from the region $u>\pi$ where the $g(u)$ series is not convergent to the complete $V(T)$ integral. The relative error increases with maturity $T$, and decreases with $\sigma_0$. The black curve corresponds to
$V_0(T$) and the colored curves correspond to $\sigma_0=0.1, 0.5, 1.0$.}
\label{Fig:err}
\end{figure}

\begin{proof}
The error bound (\ref{Verrbound}) follows from an upper bound on $g(u)$:
for any $\sigma_0>0$, we have
\begin{equation}\label{gbound}
g(u) \leq \sigma_0 \sqrt2 u \cosh(u/2)\,, \quad u \geq 0\,.
\end{equation}
It is sufficient to prove this bound for $g_0(u)$ defined by setting $h(s)\to 1$ in the
definition (\ref{galtdef}), since we have $|h(s)| \leq \frac12 \sigma_0$. 
The argument of the square root in the denominator is a concave function of $s$ and thus is bounded from below as \newline
$\cosh u - \cosh s \geq (1 - s/u) [\cosh(u)-1]$. This gives the upper bound
\begin{equation}
\int_0^u \frac{ds}{\sqrt{\cosh u -\cosh s}} \leq \frac{2u}{\sqrt{\cosh u - 1}}
\end{equation}
This yields the bound (\ref{gbound}).

The bound (\ref{gbound}) can be used to obtain an upper bound on the truncation error but the analytical result is lengthy. The simpler result (\ref{Verrbound})
is obtained from a weaker bound $g(u) \leq \sqrt2 \sigma_0 u e^{u/2}$.

\end{proof}

The bound  (\ref{Verrbound}) shows that the optimal truncation error is 
exponentially suppressed as $\sim e^{-O(1)/T}$ for $T<\pi$. 
For $T>\pi$ the error may be still small, but the
bound (\ref{Verrbound}) is not strong enough to guarantee it. 
Numerical simulations
suggest that this error bound overestimates the actual error for larger $\sigma_0$.

Numerical evaluations of the error term $V_{\rm err}(T)$ in Fig.~\ref{Fig:err}
confirm that the relative error is negligibly small for $T<1.0$ and increases rapidly
with $T$. We also note that the error decreases with $\sigma_0$. This effect is due to the factor $h(s)$ which is fast oscillating at a rate which increases with $\sigma_0$. For very large $\sigma_0$ the oscillations have an effect of suppressing the contribution from the integration region $u>\pi$, and thus decreases the error of the asymptotic series.

\section{The large $\sigma_0$ scaling limit} \label{sec:6}

In the large volatility limit $\sigma_0\to \infty$ the function $g(u)$ approaches a simple form
$\lim_{\sigma_0\to \infty} g(u) = \frac{\pi}{\sqrt2} \cosh(u/2) \equiv g_\infty(u)$.
This follows from the limit in distribution sense 
$\lim_{\epsilon\to 0} \frac{\sin(x/\epsilon)}{\pi x} = \delta_+(x)$. 
Here $\delta_+(x)$ is defined by $\int_0^u \delta_+(x) f(x) dx = \frac12 f(0)$, with $f(x)$ some test function defined on the positive axis. Taking $h(s) \to \pi \frac{\sigma_0}{2} \delta_+(s)$ into the definition (\ref{galtdef}),  the integral is trivially evaluated with the result shown.

Taking  $g(u)\to g_\infty(u)$ in (\ref{CATM}) gives $\lim_{\sigma_0\to \infty} C(K=S_0)=S_0$. (The exchange of limit and integration is justified by the Lebesgue dominated convergence theorem, using that $g(u)$ is bounded from above as shown in (\ref{gbound}).)
What is the approach to this limit? The answer to this question is related to the 
$\sigma_0\to \infty$ asymptotics of the ATM implied volatility $\sigma_{\rm BS}(0,T)$.
This asymptotics takes a simple form when considered at fixed product $\tau=\frac12 \sigma_0 T$. 
Expressed in terms of $\tau$, the Black-Scholes formula gives
\begin{equation}
S_0 - C\left( K=S_0,T= \frac{2 \tau}{\sigma_0} \right) = 2S_0  \, N\left( -\sqrt{\frac12 \sigma_0 \tau}  \,\,
\Sigma_{BS} \left( \frac{2 \tau}{\sigma_0} , \sigma_0 \right) \right)
\end{equation}
where $\Sigma_{\rm BS}(T,\sigma_0)$ is defined in Eq.~(\ref{CapSigdef}). 

The large $\sigma_0$ asymptotics of the implied volatility function $\Sigma_{BS}(T,\sigma_0)$ turns out to depend only on $\tau$, and has a calculable form, given by the following result.

\begin{proposition}
We have the limit
\begin{equation}\label{scalinglimit}
\lim_{\sigma_0\to \infty} \frac{1}{\sigma_0}\log \left[ S_0 - 
        C \left(S_0,\frac{2 \tau}{\sigma_0}\right) \right] = -\frac14 
\hat\Sigma^2_{\rm BS}(\tau) \, \tau
\end{equation}
with
\begin{equation}\label{hatCapSig}
\hat\Sigma^2_{\rm BS}(\tau) = \frac{\sin 2\lambda}{\lambda} - \frac12 ( 1 + \cos(2\lambda) )
\end{equation}
where $\lambda = \lambda(\tau)$ is the solution of 
\begin{equation}\label{2b}
\frac{\lambda}{\cos\lambda} = \tau \,.
\end{equation}
\end{proposition}

\begin{proof}
See the Appendix.
\end{proof}

The result (\ref{hatCapSig}) reproduces the asymptotic implied volatility in the
$\beta=1$ SABR model in the uncorrelated limit from \cite{PES}.

We show next that the convergence properties of the series expansion of $\hat \Sigma_{\rm BS}(\tau)$ are better behaved than expected from the non-convergence of the implied volatility $\sigma_{\rm BS}(0,T)$.
For simplicity we consider the series expansion of the implied variance
\begin{equation}\label{anseries}
\hat \Sigma_{BS}^2(\tau) = \sum_{n=0}^\infty a_n \tau^n\,.
\end{equation}
The convergence properties of this expansion are given by the following result.

\begin{proposition}\label{prop:1}
The series expansion (\ref{anseries}) 
converges for $|\tau| \leq R_\tau$ with $R_\tau = \frac{y_0}{\cosh y_0} \sim
0.662743$ and $y_0=1.19968$ the positive solution of the equation
$y \tanh y = 1$.
\end{proposition}

\begin{proof}

The function $\hat\Sigma_{\rm BS}^2(\tau)= g(\lambda(\tau))$ is the composition of two functions, with $g(\lambda)$ defined by the function on the
right-hand side of (\ref{hatCapSig}) and $\lambda(\tau)$ the inverse of $\tau(\lambda)$ in (\ref{2b}). 

The function $\lambda(\tau)$ has two branch points of order $\frac12$ on the imaginary axis at $\tau_\pm = \pm i \tau_0$ where $\tau_0 = \frac{y_0}{\cosh y_0}
\simeq 0.663$ and $y_0=1.19968$ is the positive solution of the equation $y_0 \tanh y_0 = 1$. This result follows from a study of the critical points of $f(z):= \frac{z}{\cos z}$. It is known, see e.g. Theorem 3.5.1 in \cite{Simon}, that the inversion of a complex function $w=f(z)$ around a critical point $f'(z_0)=0$, gives a multivalued function and thus $z(w)$ has a branch point at $w_0=f(z_0)$. See also
Theorem VI.6 in Flajolet and Sedgewick \cite{Flajolet2008}. For a similar application to a more complex case see Sec.~2 in \cite{Nandori2020}.

The critical points of $f$ are solutions of the equation $f'(z) = \frac{1}{\cos z} (1 + z \tan z)=0$. This equation has two types of solutions:

i) Solutions on the imaginary axis. They are at $z_\pm = \pm i y_0$ with $y_0=1.19968$ the positive solution of $y_0 \tanh y_0=1$. 
These points are mapped to $\tau_\pm=\pm i\tau_0$ with $\tau_0 = \frac{y_0}{\cosh y_0} \simeq 0.663$.

ii) Infinitely many solutions along the real axis at $z_k$ given by the solutions of
$\tan z_k = -\frac{1}{z_k}$ with $k\in \mathbb{Z}$. These solutions are mapped to $\tau_k=f(z_k)$ which are further away from origin than $\tau_0$.

In conclusion the dominant singularities of $\lambda(\tau)$ are the two branch points at $\tau_\pm = \pm i\tau_0$. 
Since $g(\lambda)$ is entire, the singularities of $\hat\Sigma_{\rm BS}^2(\tau)$
are the two branch points at $\tau_\pm = \pm i \tau_0$, which thus determine the convergence radius of the series (\ref{anseries}).

\end{proof}

While we have worked throughout with $\omega=1$, under general $\omega$, 
the expansion is in powers of $\tau = \frac12 \omega \sigma_0 T$. 
Then Prop.~\ref{prop:1} implies a corresponding convergence radius for the $T$-series expansion of the implied variance  $T_c = 1.32/(\omega \sigma_0)$.

\section{Summary and discussion}

We studied the nature of the short maturity expansion for option prices and implied volatility in the uncorrelated log-normal SABR model, and showed that in general the expansion diverges for any maturity. This implies that the series expansion is asymptotic, and that its application for numerical evaluation has to consider issues such as optimal truncation order. The optimal truncation error is
exponentially suppressed for sufficiently small maturity $\omega^2 T < \pi$.
For these maturities the first few terms of the asymptotic series give generally
a good approximation of the exact result, but for longer maturity numerical approaches are preferable to the series expansion evaluation.

Despite the asymptotic nature of the short maturity expansion in this model, it is surprising that essentially a subset of terms in the short maturity expansion of the implied volatility converges, with a finite convergence radius. 
This subset corresponds to the large $\sigma_0$ scaling limit at $\omega=1$ or more generally to the limit $\omega \to 0, \sigma_0\to \infty$ at fixed product $\omega \sigma_0$, and can be summed in closed form.

Although the analysis focused on the ATM option prices and implied volatility, the methods used are more general and may be used also for the study of the short maturity expansion at fixed strike, or in various regimes of joint small maturity-small log-strike. 
\vspace{0.2cm}

\textbf{Acknowledgements.} We thank an anonymous referee for helpful suggestions that improved the presentation of the paper. 

\clearpage

\section*{Appendix -- Large $\sigma_0$ asymptotics}
\label{sec:app}

There are two components to the asymptotics for the $\sigma_0\to \infty$ limit. 
First, we have the large $\sigma_0$ asymptotics of $g(u)$.

\begin{proposition}
The leading $\sigma_0\to \infty$ asymptotics of the function $g(u)$ is
\begin{equation}\label{ansatz}
g_\infty(u) - g(u) \simeq \frac{2\sqrt{\pi}}{\sqrt{\sigma_0 \sinh(2u)}}
\cos\left(\frac12 \sigma_0 \sinh u +\frac{\pi}{4} \right) + O(\sigma_0^{-3/2}) \,,\quad \sigma_0 \to \infty 
\end{equation}
where $g_\infty(u) \equiv \frac{\pi}{\sqrt2} \cosh(u/2)$.
\end{proposition}

\begin{proof}
Use Theorem 6.1 in Chapter 4.6 \textit{Laplace's method for contour integrals} in Olver \cite{Olver}. The leading contribution to the $\sigma_0\to \infty$ asymptotics
comes from the upper boundary of the integration region in (\ref{galtdef}).
\end{proof}

Second, the asymptotics (\ref{ansatz}) is translated into an asymptotic result for 
the integral
\begin{equation}\label{DV}
\Delta V(T, \sigma_0) = \int_0^\infty e^{-\frac{u^2}{2T}} (g(u) - g_\infty(u)) \frac{du}{\sqrt{2\pi T}}
\end{equation}
which determines the price of a covered ATM call as
\begin{equation}
S_0 - C(K=S_0,T) = \frac{2\sqrt2}{\pi} S_0 \, e^{-T/8} \Delta V(T,\sigma_0)
\end{equation}
Substituting the leading approximation (\ref{ansatz}), the integral (\ref{DV}) can be written equivalently as
\begin{eqnarray}
\Delta V(T, \sigma_0) &=& \frac{1}{\sqrt{\tau}}
\int_0^\infty e^{-\frac{\sigma_0 u^2}{4\tau}} 
\cos\left(\frac12 \sigma_0 \sinh u +\frac{\pi}{4} \right)
\frac{du}{\sqrt{\sinh(2u)}} \\
&=& \frac{1}{2\sqrt{\tau}} \int_0^\infty \left( e^{-\sigma_0 \varphi_+(u)} \sqrt{i}
+ e^{-\sigma_0 \varphi_-(u)} \sqrt{-i} \right) \frac{du}{\sqrt{\sinh(2u)}}\nonumber
\end{eqnarray}
where we denoted
\begin{equation}
\varphi_\pm(u) = \frac{u^2}{4\tau} \mp \frac{i}{2} \sinh u \,.
\end{equation}

The two integrals can be evaluated using the saddle point method. 
They are similar so it is sufficient to consider the first integral
\begin{equation}
I_+(\tau) = \int_0^\infty  e^{-\sigma_0 \varphi_+(u)} 
\frac{du}{\sqrt{\sinh(2u)}}
\end{equation}
The function $\varphi_+(u)$ has saddle points given by the solutions of the equation\\
$\varphi'_+(u)= \frac12 (\frac{u}{\tau} - i \cosh u)= 0$. This equation has solutions on the imaginary axis. The closest saddle point is at $u=i\lambda$ where $\lambda<\frac{\pi}{2}$ is the positive solution of the equation
 $\frac{\lambda}{\cos\lambda} = \tau$. This establishes (\ref{2b}). At this point we have $\varphi''_+(i\lambda) = \frac12(1/\tau + \sin\lambda)>0$.

The curves of steepest descent (ascent) from the saddle point are solutions of the equation 
\begin{equation}
\Im \varphi_+(u) = \frac{xy}{\tau} - \cos y \sinh x = 0\,,\quad u = x+ i y
\end{equation}
This equation is satisfied along the imaginary axis $x=0$, and along a curve given by
\begin{equation}
y(x) = \lambda\left(\tau \frac{\sinh x}{x} \right)
\end{equation}
where $\lambda(\tau)$ is the solution of $\frac{\lambda}{\cos \lambda} = \tau$.
Since $\lim_{\tau\to \infty} \lambda(\tau) = \frac{\pi}{2}$, the curve $y(x)$ approaches $\pi/2$ as $|x| \to \infty$.
$y(x)$ intersects the imaginary axis at the saddle point $S=i\lambda$. See Fig.~\ref{Fig:paths} for an illustration for $\tau=1$.

For the application of the saddle point method we deform the contour of integration from the real positive axis such that it passes through the saddle point $S$ as shown in Fig.~\ref{Fig:paths} as the red curve. 
Along the vertical segment $u\in [0,S]$ the function $\Re\varphi_+(u)$ increases as $u\to S$ and along the curve in the first quadrant, $\Re\varphi_+(u)$ increases further as we move away from the saddle point. We call the latter curve a steepest descent path, and the former a steepest ascent path.

The integration contour can be deformed from the real axis to this path without encountering any singularities of the integrand $1/\sqrt{\sinh(2u)}$.
The closest singularities of this function to the origin are branch points at 
$\pm i\frac{\pi}{2}$, which are farther away than the saddle point.

\begin{figure}[h]
\centering
\includegraphics[width=2.0in]{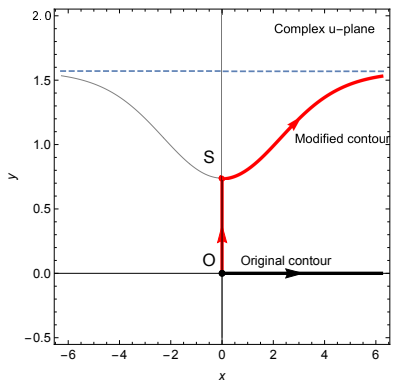}
\caption{Steepest descent paths $\Im \varphi_+(u)=0$ for $\tau=1.0$. 
The integration contour is deformed from the real axis to the path shown in red. Along the vertical piece of this path the function $\Re \varphi_+(u)$ increases as one approaches $S$, and along the curved portion it increases further as one moves away from the saddle point.
The contribution from the vertical path along the imaginary axis cancels out in the final result. The only contribution appears from the curved path.}
\label{Fig:paths}
\end{figure}

The integral is written as a sum of two contributions from the two pieces of the 
contour $I_+ = \int_0^S  + \int_S^{i\frac{\pi}{2}+\infty}$ where $S=i\lambda$ is the saddle point. The first term is imaginary, and cancels against an identical contribution from $I_-$. The second integral is dominated by the contribution of the saddle point and is expressed as a Laplace integral by introducing the new integration variable $\zeta = \varphi_+(u) - \varphi_+(S)$. Since along the contour $\Im \varphi_+(u)=0$, the variable $\zeta$ is real. 

Expanding the integrand as $\zeta \to 0$ we obtain
\begin{eqnarray}\label{int2}
&&\int_S^{i\frac{\pi}{2}+\infty} e^{-\sigma_0 \varphi_+(u)} 
\frac{du}{\sqrt{\sinh(2u)}} =
 e^{-\sigma_0 \varphi_+(i\lambda)} \int_0^\infty e^{-\sigma_0 \zeta}
\frac{d\zeta}{\sqrt{\sinh (2u)} \varphi'_+(u)} \\
&& = e^{-\sigma_0 \varphi_+(i\lambda)}\frac{1}{\sqrt{2i\varphi''_+(S)\sin(2\lambda)}}
\int_0^\infty e^{-\sigma_0 \zeta} \frac{d\zeta}{\sqrt{\zeta}} (1 + O(\sqrt{\zeta}))
\nonumber
\end{eqnarray}
The integral can be evaluated term by term by Watson's lemma. The leading contribution is
\begin{equation}
\int_0^\infty e^{-\sigma_0 \zeta}
\frac{d\zeta}{\sqrt{\sinh (2u)} \varphi'_+(u)} = \sqrt{\frac{\pi}{2i\varphi''_+(S)\sin(2\lambda)}}
\cdot \frac{1}{\sqrt{\sigma_0}} (1 + O(\sigma_0^{-1/2}))
\end{equation}

Collecting all factors gives 
\begin{equation}\label{Isol}
\Re[I_+ \sqrt{i}] = C
\frac{1}{\sqrt{\sigma_0}} e^{-\sigma_0 \varphi_+(i\lambda)} (1 + O(\sigma_0^{-1/2}))
\end{equation}
with $\varphi_+(i\lambda) = \frac14(2\sin\lambda - \lambda \cos\lambda)$ and
\begin{equation}
C= \sqrt{\frac{\pi}{2 \varphi''_+(i\lambda) \sin(2\lambda)}} =
\sqrt{\frac{\pi \lambda}{\sin\lambda \cos^2\lambda(1 + \lambda \tan\lambda)}}
\end{equation}

The same result is obtained for $I_-$. Adding their contributions reproduces the exponential factor implicitly defined by the $\log$ arg in (\ref{scalinglimit}).

\end{document}